\documentclass[journal]{IEEEtran}

\newif\ifdraft
\drafttrue

\usepackage{adjustbox}
\usepackage{blindtext}
\usepackage{graphicx}
\usepackage[utf8]{inputenc}
\usepackage{chngpage}
\usepackage{graphics} 
\usepackage{epsfig} 
\usepackage{mathptmx} 
\usepackage{amsmath} 
\usepackage{amssymb}  
\usepackage{cite}
\usepackage{multicol}
\usepackage{mathtools, cuted}
\usepackage[cmintegrals]{newtxmath}
\usepackage{epstopdf}
\usepackage{graphics} 
\usepackage{epsfig} 
\usepackage{subfigure}

\usepackage{titlesec}
\usepackage{epsfig} 
\usepackage{mathptmx} 
\usepackage{amsmath} 
\usepackage{blkarray, bigstrut}
\usepackage{amssymb}  
\usepackage{cite}
\usepackage{multicol}
\usepackage{mathtools, cuted}
\usepackage[cmintegrals]{newtxmath}
\usepackage{makeidx}         

\usepackage{algorithmic}

\usepackage{multicol}        
\usepackage[bottom]{footmisc}
\usepackage{caption}

\usepackage{nomencl}
\usepackage[usenames, dvipsnames]{color}

\usepackage{mathtools, cuted}
\usepackage{lipsum, color}

 \usepackage{amsthm}
 \theoremstyle{definition}
\newtheorem{definition}{Definition}

 \theoremstyle{property}
\newtheorem{property}{Property}

\theoremstyle{theorem}
\newtheorem{theorem}{Theorem}

\theoremstyle{assumption}
\newtheorem{assumption}{Assumption}

\theoremstyle{lemma}

\theoremstyle{remark}

 \theoremstyle{proposition}

\usepackage{algorithmic}    

\usepackage[ruled,vlined]{algorithm2e}

\usepackage{array}
\newcolumntype{M}[1]{>{\centering\arraybackslash}m{#1}}
\usepackage{float}
\usepackage{graphicx}

\makeatletter
\newcommand*{\rom}[1]{\expandafter\@slowromancap\romannumeral #1@}
\makeatother

\ifCLASSINFOpdf
\else
\fi
\begin{document}
%
\title{A Concurrent Switching Model for\\Traffic Congestion  Control}

\author{Hossein Rastgoftar, Xun Liu, and Jean-Baptiste Jeannin
\thanks{H. Rastgoftar is with the Department of Aerospace and Mechanical Engineering at the University of Arizona, Tucson, AZ 85721, USA {\tt\small hrastgoftar@arizona.edu}}
\thanks{X. Liu is with the Department of Mechanical Engineering at Villanova University, Villanova, PA 19085, USA {\tt\small xliu8@villanova.edu}}
\thanks{J.-B. Jeannin is with the Department
of Aerospace Engineering, University of Michigan, Ann Arbor,
MI, 48109 USA {\tt\small jeannin@umich.edu}.}

}
\markboth{}%
{Shell \MakeLowercase{\textit{et al.}}: Bare Demo of IEEEtran.cls for IEEE Journals}
%



\maketitle

\begin{abstract}
%

We introduce a new conservation-based approach
for traffic coordination modeling and control in a network of
interconnected roads (NOIR) with switching movement phase rotations at every NOIR junction. For modeling of traffic evolution,
we first assume that the  movement phase rotation  is cyclic at
every NOIR junction, but the duration of each movement phase can
be arbitrarily commanded by traffic signals. Then, we propose a novel concurrent
switching dynamics (CSD) with deterministic transitions among a finite number of
states, representing the NOIR movement phases. We define the
CSD control as a cyclic receding horizon optimization problem
with periodic quadratic cost 
and constraints. More specifically,
the cost is defined so that the traffic density is minimized and
the boundary inflow is uniformly distributed over the boundary
inlet roads, whereas the cost parameters are periodically changed
with time. The constraints are linear and imposed by a trapezoidal fundamental diagram at every NOIR road so that traffic
feasibility is assured and traffic over-saturation is avoided. The
success of the proposed traffic boundary control is demonstrated
by simulation of traffic congestion control in Downtown Phoenix (see Fig. \ref{realmap}).
\end{abstract}




\section{Introduction}


Traffic congestion is a prevalent global phenomenon that arose accompanied by the urbanization process, which imposes enormous costs on both economy and ecology. According to the statistical investigation, due to the traffic congestion the average annual cost for a driver in the US was $97$ hours and $\$1,348$ in 2018\cite{reed2019global}. To this end, traffic management has been extensively studied by scholars in order to exploit the capacity of the existing road network such that the congestion can be alleviated without significant cost. Over the past decades, A number of approaches for traffic congestion management have been proposed, which can be roughly categorized into two groups: physics-based approaches and  light-based approaches.

Physics-based approaches refer to the methods depending on the traffic model which are related with traffic flow and queue theory. Incorporating with the mass conservation law, the Cell Transmission Model (CTM) is widely applied to spatially partition a network of interconnected roads (NOIR) into road elements in the process of physics-based traffic coordination modeling\cite{DAGANZO1995cell}. Based on the CTM theory, Ba and Savla\cite{ba2016ondistributed} propose an optimal control method to achieve the optimal traffic flow in consideration of the traffic density of the network. Also, to analyze and control the traffic dynamics, Geroliminis and Daganzo \cite{GEROLIMINIS2008existence} introduce the concept of macroscopic fundamental diagrams (MFD), which prove the relationship between the average traffic flow and the traffic density for urban-scale network. Since then, researchers have conducted a number of studies to integrate the MFD theory with the CTM approach, in order to enhance the accuracy of the large-scale traffic coordination model~\cite{munoz2003traffic,yang2017fundamental}. Furthermore, Haddad et al.\cite{Geroliminis2013optimal,HADDAD2017optimal} incorporate the perimeter control approach with the MFD theory to obtain the optimal flow of the traffic zone. Li et al.~\cite{li2017traffic} introduce a feedback control strategy based on the MFD model to maximize the traffic volume of the network. In addition, model predictive control (MPC) is another prevalent approach for physics-based traffic dynamics optimization\cite{Liu2021ConservationBased,Jamshidnejad2018sustainable,Rastgoftar2020integrative}. To overcome the nonlinearity of the prediction model, Lin et al.\cite{lin2011fast} rewrite the nonlinear MPC model into a mixed-integer linear programming (MILP) optimization problem and adopt the efficient MILP solver to guarantee the global optimum of the traffic flow.

Light-based approaches attempt to mitigate the traffic pressure by optimizing the timing plan of the traffic signal using model-free strategies, such as fuzzy logic, Genetic Algorithm (GA), Markov Decision Process (MDP), and neural network etc. Chiu and Chand~\cite{Chiu1993Adaptive} adopt the fuzzy-based strategy to control multiple intersections in a network of two-way streets with no turning motions. By collecting and processing the local traffic data, the optimization of the signal cycle length is achieved according to the degree of saturation at each intersection. Wei et al.\cite{Wu2001Traffic} introduce a fuzzy logic approach to optimize an isolated traffic intersection with four approaches and four phases. Adjustments to the signal timing are made in response to different user's demand for green time. Also, researchers employ the GA\cite{Sanchez2010Traffic} and MDP \cite{Yin2014Traffic,haijema2008MDP} approaches in the traffic signal optimization problem to reduce the traffic delay. Furthermore, with the prompt development of Artificial Intelligence (AI) technology and the breakthrough of computational capacity, Reinforcement Learning (RL) is becoming an increasingly popular data-driven approach applied for the optimization of the traffic signal plan\cite{abdulhai2003reinforcementlearning,Zeng2018Adaptive,Tantawy2012Reinforcement}. The successful application of the RL algorithm on the optimization of the traffic signal plan presents the ability of the approach to learn through dynamic interaction with the environment. Moreover, the Deep Learning (DL) integrated RL algorithm, which is widely known as  Deep Reinforcement Learning (DRL), is proposed to improve the applicability of the learning algorithm on the traffic signal optimization of the large-scale NOIR\cite{Wu2020Multi-Agent,Zhao2019Traffic,Rasheed2020Deep,Chu2020Deep}.

This paper considers the problem of modeling and control of  traffic evolution in an NOIR with movement phase rotations included at every junction. Compared with our previous work\cite{Rastgoftar2021APF}, the main goal is to model traffic evolution as a system with cyclic and deterministic transitions over a finite number of states representing NOIR movement phases. To this end, we assume that the movement phase rotation is periodic at every junction, but the durations of movement phases are not necessarily the same. To overcome this complexity,
we propose to replace ``movement phase duration'' by ``movement phase repetition'' at every NOIR junction. To this end, we use  a cycle graph with the nodes representing the movement phases, and the edges specifying transitions from the current movement phases to the next ones. Note that the cycle graph authorizes transitions to the next movement phase, which is either the same or different than the current movement phase. For traffic congestion, we assign optimal boundary inflow by solving a constrained receding horizon optimization problem with a quadratic cost function that periodically changes with time. The constraints of this boundary control problems imposing the feasibility of traffic evolution are linear and obtained by using a trapezoid fundamental diagram. Therefore, the optimal boundary inflow is obtained by solving a quadratic programming problem at every discrete time $k\in \mathbb{N}$.

This paper is organized as follows: The definitions and topology of traffic network are presented in Section \ref{Traffic Network}. The Problem Statement and Specification are presented in Sections \ref{sec:problem} and  \ref{Problem Specification}, respectively, and are followed by the development of the traffic network dynamics in Section \ref{Traffic Network Dynamics}. Traffic congestion control is presented as a periodic receding horizon optimization problem in Section \ref{Traffic Congestion Control}. Simulation results are presented in Section \ref{Simulation Results}, followed by Conclusion in Section \ref{Conclusion}.




\section{Traffic Network}\label{Traffic Network}
We consider a NOIR with unidirectional roads defined by set $\mathcal{V}$ and junctions defined by set $\mathcal{W}$. Interconnections between the roads are specified by graph $\mathcal{G}\left(\mathcal{V},\mathcal{E}\right)$ with node set $\mathcal{V}$ and edge set $\mathcal{E}\subset \mathcal{V}\times \mathcal{V}$.
Note that the set of nodes in the graph represents the set of unidirectional roads in the NOIR.
In this paper, edges defined by $\mathcal{E}$ satisfy the following two properties:
\begin{property}
If $\left(i,j\right)\in \mathcal{E}$, then, $\left(j,i\right)\notin \mathcal{E}$.
\end{property}
\begin{property}
If $\left(i,j\right)\in \mathcal{E}$, then, traffic is directed from $i\in \mathcal{V}$ towards $j\in \mathcal{V}$ which in turn implies that $i\in \mathcal{V}$ is the upstream road.
\end{property}
For every $i\in \mathcal{V}$, $\mathcal{I}_i=\left\{j\in \mathcal{V}:\left(j,i\right)\in \mathcal{E}\right\}$ and  $\mathcal{O}_i=\left\{j\in \mathcal{V}:\left(i,j\right)\in \mathcal{E}\right\}$ define the in-neighbors and out-neighbors of $i\in \mathcal{V}$, respectively. In particular, the following conditions hold:
\begin{enumerate}
    \item Traffic directed from $j\in \mathcal{I}_i$ towards $i\in \mathcal{V}$, if $\mathcal{I}_i\neq \emptyset$.
    \item Traffic directed from $i\in \mathcal{V}$ towards  $j\in \mathcal{O}_i$, if $\mathcal{O}_i\neq \emptyset$.
\end{enumerate}
By knowing edge set $\mathcal{E}$, we can express $\mathcal{V}$ as $\mathcal{V}=\mathcal{V}_{in}\bigcup \mathcal{V}_{out}\bigcup\mathcal{V}_I$ where inlet road set $\mathcal{V}_{in}$, outlet  road set $\mathcal{V}_{out}$, and interior road set $\mathcal{V}_I$ are formally defined as follows: 
\begin{subequations}
\begin{equation}
    \mathcal{V}_{in}=\left\{i\in \mathcal{V}:\mathcal{I}_i=\emptyset\right\},
\end{equation}
\begin{equation}
    \mathcal{V}_{out}=\left\{i\in \mathcal{V}:\mathcal{O}_i=\emptyset\right\},
\end{equation}
\begin{equation}
    \mathcal{V}_{I}=\left\{i\in \mathcal{V}:\mathcal{I}_i\neq\emptyset,~\mathcal{O}_i\neq\emptyset\right\}.
\end{equation}
\end{subequations}
Assuming the NOIR has $m$ junctions, $\mathcal{W}=\left\{1,\cdots,m\right\}$ defines the junctions' identification numbers. At junction $i\in \mathcal{W}$, the movement phase rotation 
is defined by cycle graph $\mathcal{C}_i\left(\mathcal{E}_i,\mathcal{R}_i\right)$, where $\mathcal{E}_i\subset \mathcal{E}$ and $\mathcal{R}_i\subset \mathcal{E}_i\times \mathcal{E}_i$ defines nodes and edges of cycle graph $\mathcal{C}_i$. Set $\mathcal{E}_i$ can be expressed as 
\begin{equation}
    \mathcal{E}_i=\bigcup_{j=1}^{r_i}\mathcal{E}_{i,j},\qquad \forall i\in \mathcal{W}
\end{equation}
For better clarification of the above definitions, we consider an example NOIR of Phoenix City shown in Fig.~\ref{realmap} with $60$ 
 unidirectional roads
 identified by set $\mathcal{V}=\{1,\cdots,60\}$, where $\mathcal{V}=\mathcal{V}_{in}\bigcup\mathcal{V}_{out}\bigcup\mathcal{V}_I$, $\mathcal{V}_{in}=\left\{1,\cdots,11\right\}$, $\mathcal{V}_{out}=\left\{12,\cdots,22\right\}$, and $\mathcal{V}_I=\left\{23,\cdots,60\right\}$. The NOIR shown in Fig.~\ref{realmap} consists of $14$ junctions defined by $\mathcal{W}=\{1,\cdots,14\}$. Set $\mathcal{E}_{12}=\mathcal{E}_{12,1}\bigcup \mathcal{E}_{12,2}\bigcup \mathcal{E}_{12,3}\bigcup \mathcal{E}_{12,4}$ defines the four movement phases at junction $12\in \mathcal{W}$, where $\mathcal{E}_{12,1}=\{(2,27), (2,35), (2,55)\}\subset \mathcal{E}$, $\mathcal{E}_{12,2}=\{(36,35), (36,27), (36,19), (36,55)\}\subset \mathcal{E}$, $\mathcal{E}_{12,3}=\{(54,19), (54,27), (54,55), (54,35)\}\subset \mathcal{E}$, and $\mathcal{E}_{12,4}=\{(46,19), (46,55), (46,35), (46,27)\}\subset \mathcal{E}$. The four possible movement phases at junction $12$ are shown in Fig.~\ref{fig:example_junction_phases}. 

\begin{assumption}\label{repeatedmovephase}
The next movement phase at junction $i\in \mathcal{V}$ can be either the same as or different with the current movement phase. Mathematically, The next movement phase $\mathcal{E}_{i,h}$ is not necessarily different with the current movement phase $\mathcal{E}_{i,l}$, if $\left(\mathcal{E}_{i,l},\mathcal{E}_{i,h}\right)\in \mathcal{R}_i$. 
\end{assumption}
 Assumption \ref{repeatedmovephase} implies that the  phase rotation cycle is still advanced, if next movement phase $\mathcal{E}_{i,h}$ is either the same as or different with current movement phase $\mathcal{E}_{i,l}$,
\begin{assumption}\label{concurrent}
Movement phase rotation occurs concurrently across the NOIR junctions.  
\end{assumption}
Assumption \ref{concurrent} is not a restricting assumption since repetition of  movement phases is authorized at a junction per Assumption  \ref{repeatedmovephase}. Indeed the duration of particular movement phase $\mathcal{E}_{i,j}$ ($j=1,\cdots,r_i$) can be chosen arbitrarily large through defining edge set $\mathcal{R}_i$ for $i\in \mathcal{W}$. As shown in Fig. \ref{CNOIR}, repetition of a particular movement phase denoted by $\mathcal{E}_{i,j}$, is authorized by defining $\mathcal{E}_{i,k}=\mathcal{E}_{i,j}$ and  $\mathcal{E}_{i,l}=\mathcal{E}_{i,j}$ where $\left(\mathcal{E}_{i,j},\mathcal{E}_{i,k}\right)\in \mathcal{R}_i$ and $\left(\mathcal{E}_{i,k},\mathcal{E}_{i,l}\right)\in \mathcal{R}_i$.
\begin{figure}[ht]
    \centering
    \includegraphics[width=\linewidth]{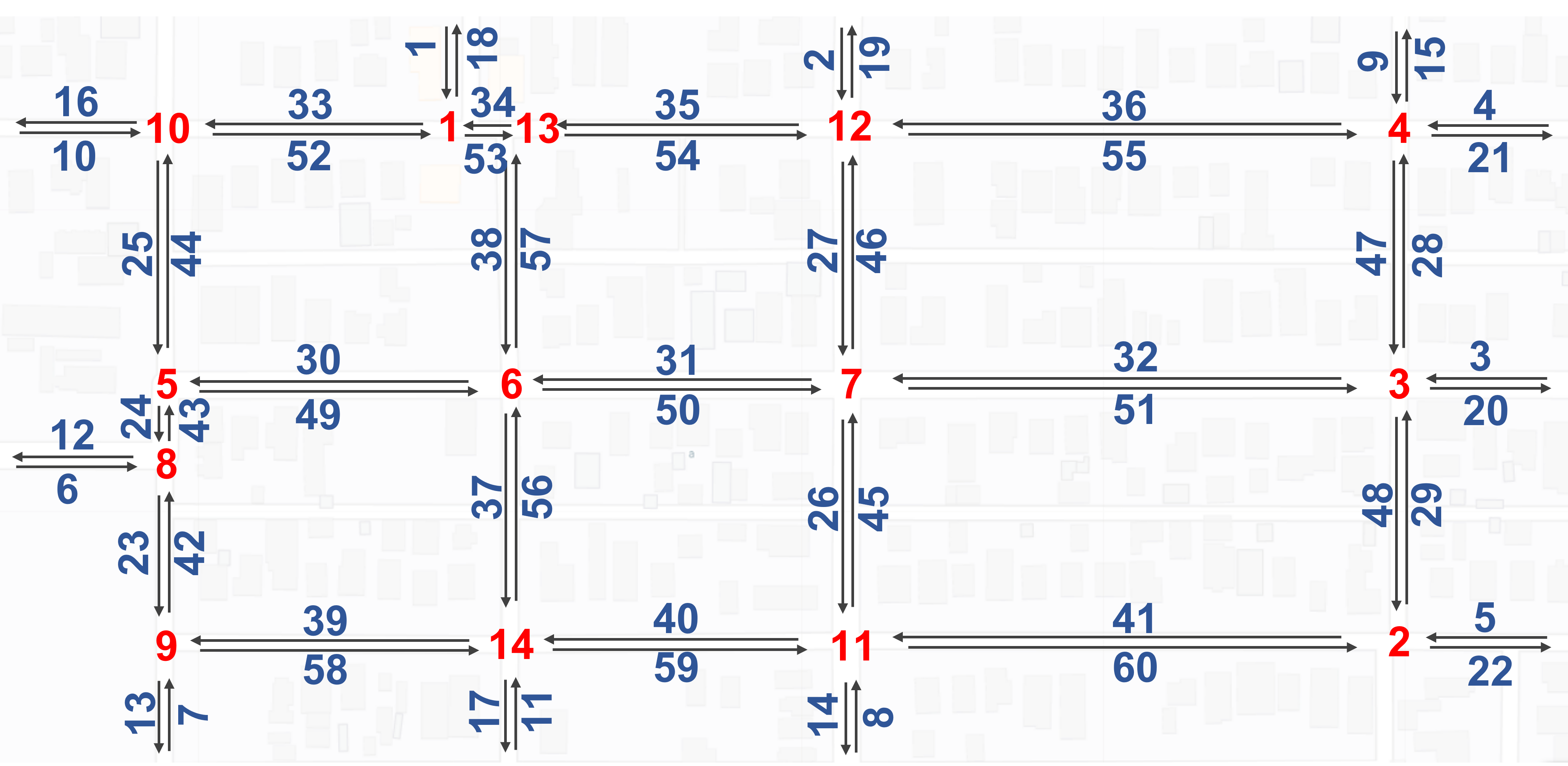}
    \vspace{-.3cm}
    \caption{Example NOIR: Street map of Phoenix.
    Numbers in blue denote the unidirectional roads $\mathcal{V}=\{1,\cdots,60\}$, and red color numbers represent junctions $\mathcal{W}=\{1,\cdots,14\}$.}
    \label{realmap}
\end{figure}


\begin{definition}
The NOIR movement phase rotation is cyclic and defined by graph $\mathcal{C}_{\mathrm{NOIR}}\left(\mathcal{L},\mathcal{M}\right)$ with node set
\[
\mathcal{L}=\mathcal{E}_1\times \cdots \times\mathcal{E}_m
\]
and edge set $\mathcal{M}\subset \mathcal{L}\times \mathcal{L}$, where $\times$ is the Cartesian product symbol. 
\end{definition}

\begin{theorem}\label{thm1}
Let $r_i$ be the number of movement phases at junction $i\in \mathcal{W}=\left\{1,\cdots,m\right\}$, and movement phase rotation be cyclic and satisfy condition  \ref{rotation} at every junction $i\in\mathcal{W}=\left\{1,\cdots,m\right\}$. Then, the NOIR cycle is completed in $n_c$ time steps where 
\begin{equation}
    n_c=\mathrm{lcm}\left(r_1,\cdots,r_m\right)
\end{equation}
is the lowest common multiple of $r_1$, $\cdots$, $r_{m}$.
\end{theorem}
\begin{proof}
Completion of movement phase rotation at every junction is deterministic and independent of other junctions at every junction $i\in \mathcal{W}$. By imposing Assumption \ref{concurrent}, The edges of graph $\mathcal{C}_{\mathrm{NOIR}}\left(\mathcal{L},\mathcal{M}\right)$, defined by $\mathcal{L}$, are restricted to satisfy the following constraints:
\begin{equation}\label{rotation}
    \bigwedge_{i\in \mathcal{W}}\bigwedge_{j=1}^{r_i}\bigwedge_{h=1}^{r_i}
    \left(\left(\mathcal{E}_{i,j},\mathcal{E}_{i,h}\right)\in \mathcal{R}_i\right).
\end{equation}
Because movement phase rotation, completed in ${r}_i$ time steps, is independent at every junction $i\in \mathcal{W}$, the cycle of graph $\mathcal{C}_{\mathrm{NOIR}}$ is completed in $n_c$ time steps where $n_c$ is the lowest common multiple of $r_1$, $\cdots$, $r_{m}$ and obtained by \eqref{rotation}.
\end{proof}
Per Theorem \ref{thm1}, graph $\mathcal{C}_{\mathrm{NOIR}}\left(\mathcal{L},\mathcal{M}\right)$ consists of $n_c$ nodes defined by set 
\begin{equation}
\mathcal{L}=\left\{\lambda_j=\left(\mathcal{E}_{1,j},\cdots,\mathcal{E}_{i,j},\mathcal{E}_{m,j}\right):i\in \mathcal{W},~j=0,1,\cdots,n_c-1\right\}
\end{equation}
\begin{definition}
The identification number of the NOIR movement phases are defined by set 
\begin{equation}
    \sigma=\left\{0,\cdots,n_c-1\right\}.
\end{equation}
\end{definition}
Fig. \ref{fig:phoenix_noir} shows an example graph $\mathcal{C}_{\mathrm{NOIR}}$ specifying the NOIR movement phase rotations for a traffic network with $m$ junctions.
\begin{figure}[ht]
    \centering
    \includegraphics[width=\linewidth]{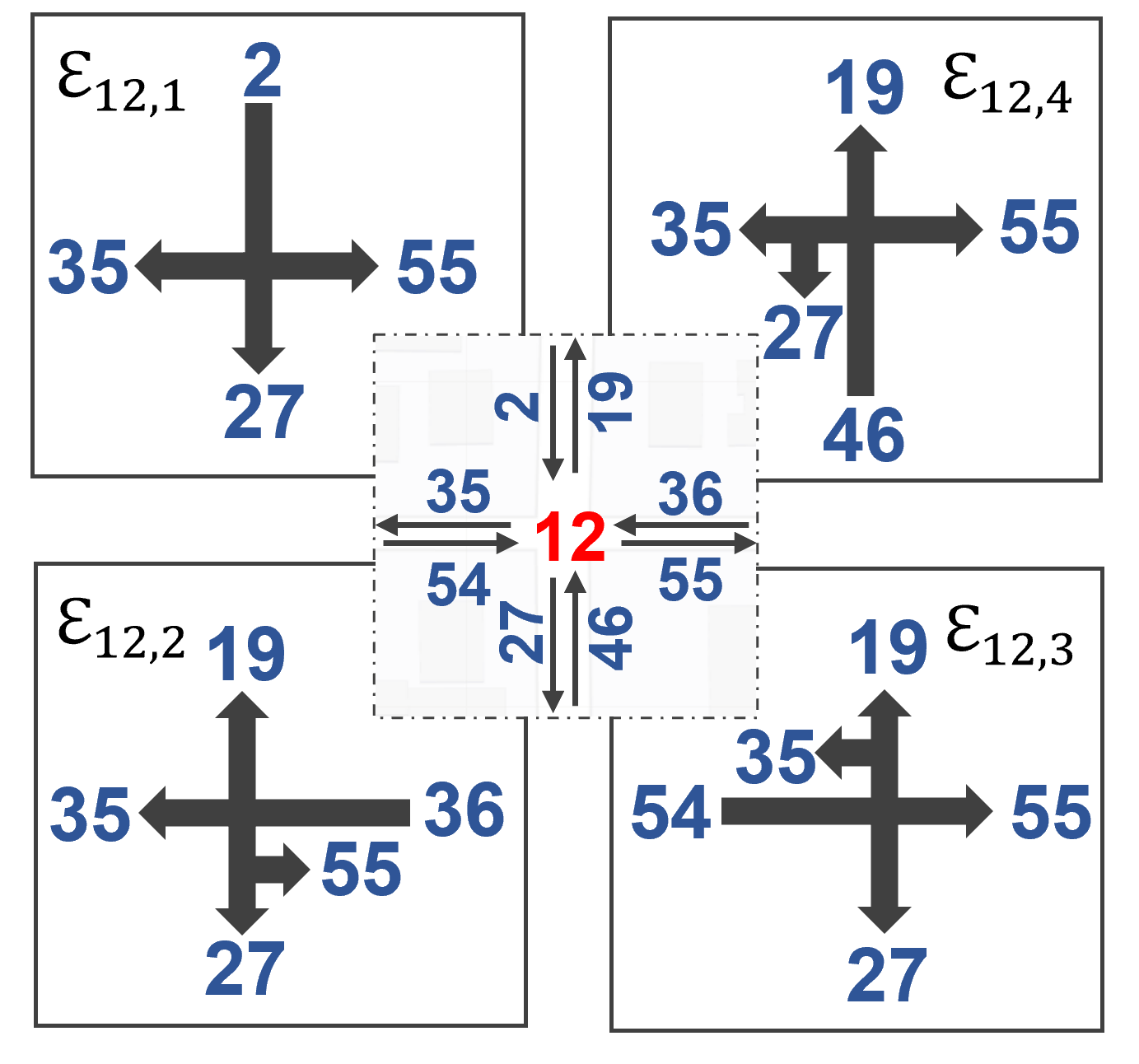}
    \vspace{-.3cm}
    \caption{Possible movement phases at junction $12\in\mathcal{W}$}
    \label{fig:example_junction_phases}
\end{figure}

\begin{figure}[ht]
    \centering
    \includegraphics[width=\linewidth]{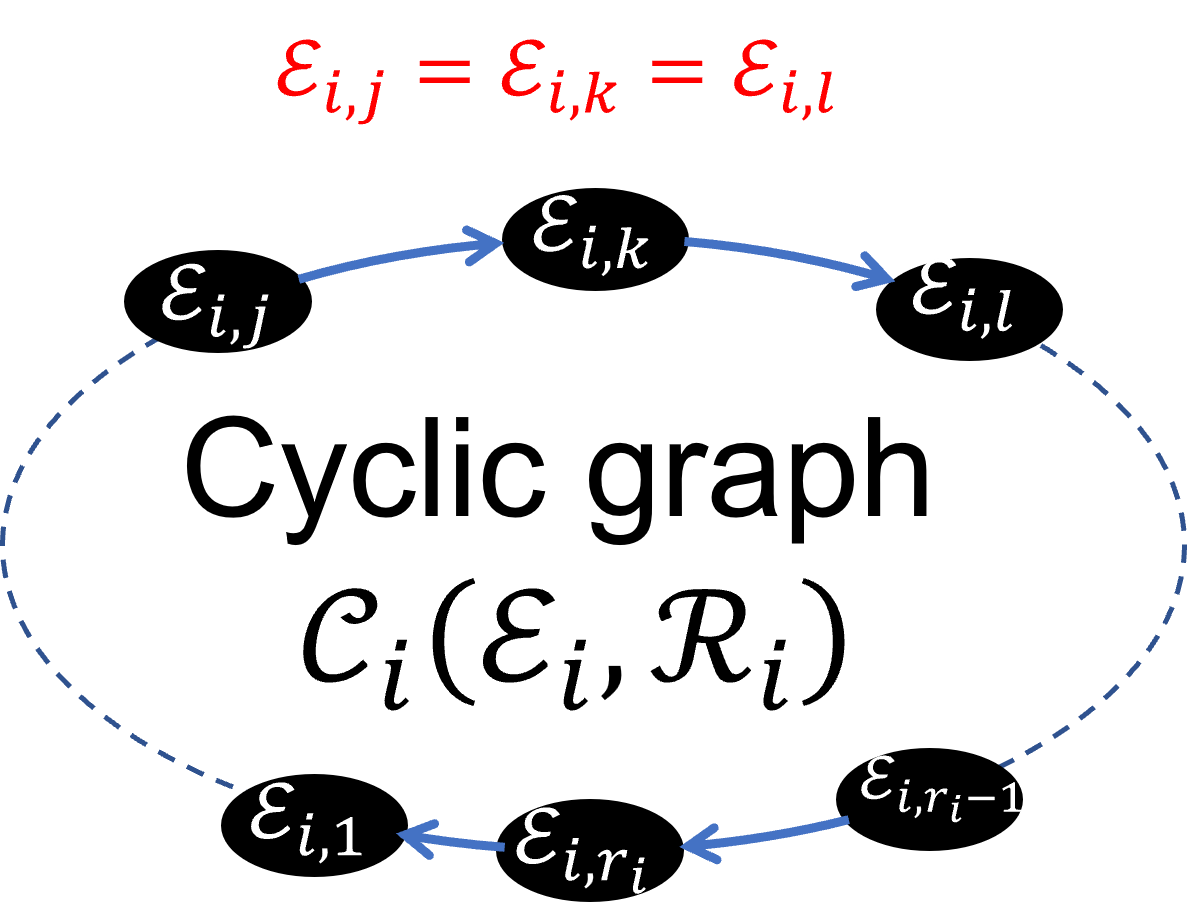}
    \vspace{-.3cm}
    \caption{Example NOIR: Street map of Phoenix}
    \label{CNOIR}
\end{figure}

\begin{figure}[htb]
\centering
\includegraphics[width=3.3  in]{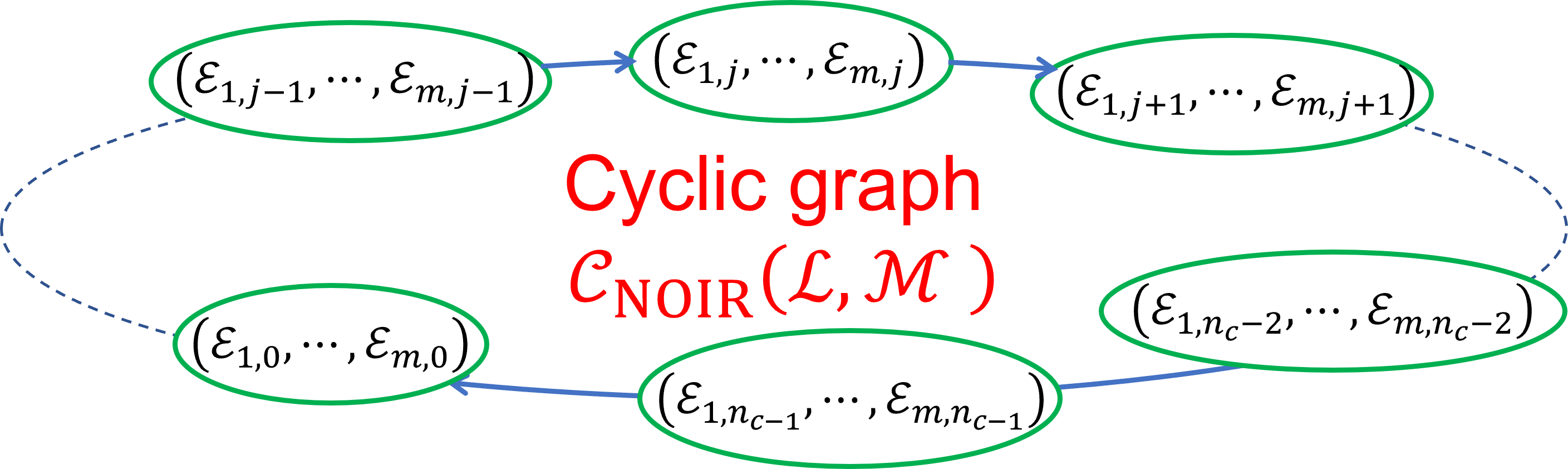}
\vspace{-.3cm}
\caption{Schematic of the cyclic graph $\mathcal{C}_i$ with movement phase repetition. }
\label{fig:phoenix_noir}
\end{figure}

\section{Problem Statement}
\label{sec:problem}
Before proceeding to state the problem studied in this paper, we define variables $\zeta[k]\in \sigma$ and $\gamma[k]\in \sigma$ by
\begin{align*}
\zeta[k] & =k\bmod n_c\\
\gamma[k] &= (k+1)\bmod n_c = \zeta[k]+1\bmod n_c
\end{align*}
%
%
We apply the cell transmission model to model traffic in a NOIR by 
\begin{equation}
\label{eq:update}
    \rho_i[k+1]=\rho_i[k]+y_i[k]-z_i[k]+u_i[k],
\end{equation}
where $\rho_i$ is the traffic density; $u_i$ is the boundary inflow; and $z_i$ and $y_i$ are the  network outflow and network inflow of road $i\in \mathcal{V}$, respectively; and they are defined as follows:
\begin{subequations}
\begin{equation}
    z_i[k]=p_{i,\lambda_{\zeta[k]}}\rho_i[k],\qquad  \lambda_{\zeta[k]} \in \mathcal{L},~\zeta[k]\in\sigma 
\end{equation}
\begin{equation}
\begin{split}
    y_i[k]=&\sum_{j\in \mathcal{O}_i}q_{j,i,\lambda_{\zeta[k]}}z_j[k]\\
    =&\sum_{j\in \mathcal{O}_i}q_{j,i,\lambda_{\zeta[k]}}p_{j,\lambda_{\zeta[k]}}\rho_j[k]
\end{split}
,\qquad  \lambda_{\zeta[k]} \in \mathcal{L},~\zeta[k]\in\sigma 
\end{equation}
\end{subequations}
at every discrete time $k$. Note that 
\begin{equation}\label{zetagamma}
\left(\lambda_{\zeta[k]},\lambda_{\gamma[k]}\right)\in \mathcal{M}, \qquad \forall k,
\end{equation}
where $p_{i,\lambda_\zeta[k]}\in \left(0,1\right]$ is the outflow probability of road $i\in \mathcal{V}$ at discrete time $k$ when $\lambda_{\zeta[k]}\in \mathcal{M}$ is the active NOIR movement phase.
Also, $q_{j,i,\lambda_\zeta[k]}\in \left(0,1\right]$ assigns the fraction of outflow of road $i\in \mathcal{V}$ directed towards $j\in \mathcal{O}_i$ under NOIR movement phase $\lambda_{\zeta[k]}\in \mathcal{M}$ at discrete time $k$.

To assure the traffic feasibility, traffic density of road $i\in \mathcal{V}$ must satisfy the following inequality constraints:
\begin{subequations}\label{CCC}
\begin{equation}\label{C1}
    \bigwedge_{i\in \mathcal{V}_{in}}\left(u_i[k]\geq 0\right),\qquad k\in \mathbb{N},
\end{equation}
\begin{equation}\label{C2}
    \sum_{i\in \mathcal{V}_{in}}u_i[k]=u_0,\qquad k\in \mathbb{N},
\end{equation}
\begin{equation}\label{C3}
    \bigwedge_{i\in \mathcal{V}}\left(0\leq \rho_i[k]\leq \bar{\rho}_{\max} \right),\qquad k\in \mathbb{N}.
\end{equation}
\end{subequations}
Condition \eqref{C1} ensures that the boundary inflow is non-negative at every discrete time $k$. Condition \eqref{C2}, requiring the net boundary inflow is equal to constant value $u_0$, is valid when the demand for using the NOIR is high. 
Condition \eqref{C3} assures that the traffic density is always positive and does not exceed $\rho_{\max}$.
We use the Fundamental Diagram (FD)~\cite{WU2011empirical} with the schematic shown in Fig.~\ref{Fundamentaldiagram}  to assure feasibility of the network outflow at every road $i\in\mathcal{V}$. As shown in Fig.~\ref{Fundamentaldiagram}, the fundamental diagram is a trapezoid that is determined by knowing $\bar{\rho}_{\min}$, $\bar{\rho}_{\rm{mid}}$, $\bar{\rho}_{\max}$, and $\bar{z}_{\max}$. In particular, the FD is applied to assure that the outflow of road $i\in \mathcal{V}$, denoted by $z_i$ is feasible by satisfying the following inequality constraints:
\begin{subequations}\label{fnd_all}
\begin{equation}\label{fnd1}
    \bigwedge_{i\in \mathcal{V}}\left(z_i[k]\geq 0\right),\qquad k\in \mathbb{N},
\end{equation}
\begin{equation}\label{fnd2}
    \bigwedge_{i\in \mathcal{V}}\left(z_i[k]\leq {\bar{z}_{\max}\rho_i[k]\over \bar{\rho}_{\min}}\right),\qquad k\in \mathbb{N},
\end{equation}
\begin{equation}\label{fnd3}
    \bigwedge_{i\in \mathcal{V}}\left(z_i[k]\leq \bar{z}_{\max}\right),\qquad k\in \mathbb{N},
\end{equation}
\begin{equation}\label{fnd4}
    \bigwedge_{i\in \mathcal{V}}\left(z_i[k]\leq \frac{\bar{z}_{\max}\left(\rho_i[k]-\bar{\rho}_{\max}\right)} {(\bar{\rho}_{\rm{mid}}-\bar{\rho}_{\max})}\right),\qquad k\in \mathbb{N}.
\end{equation}
\end{subequations}


\begin{figure}[htb]
\centering
\includegraphics[width=3.3  in]{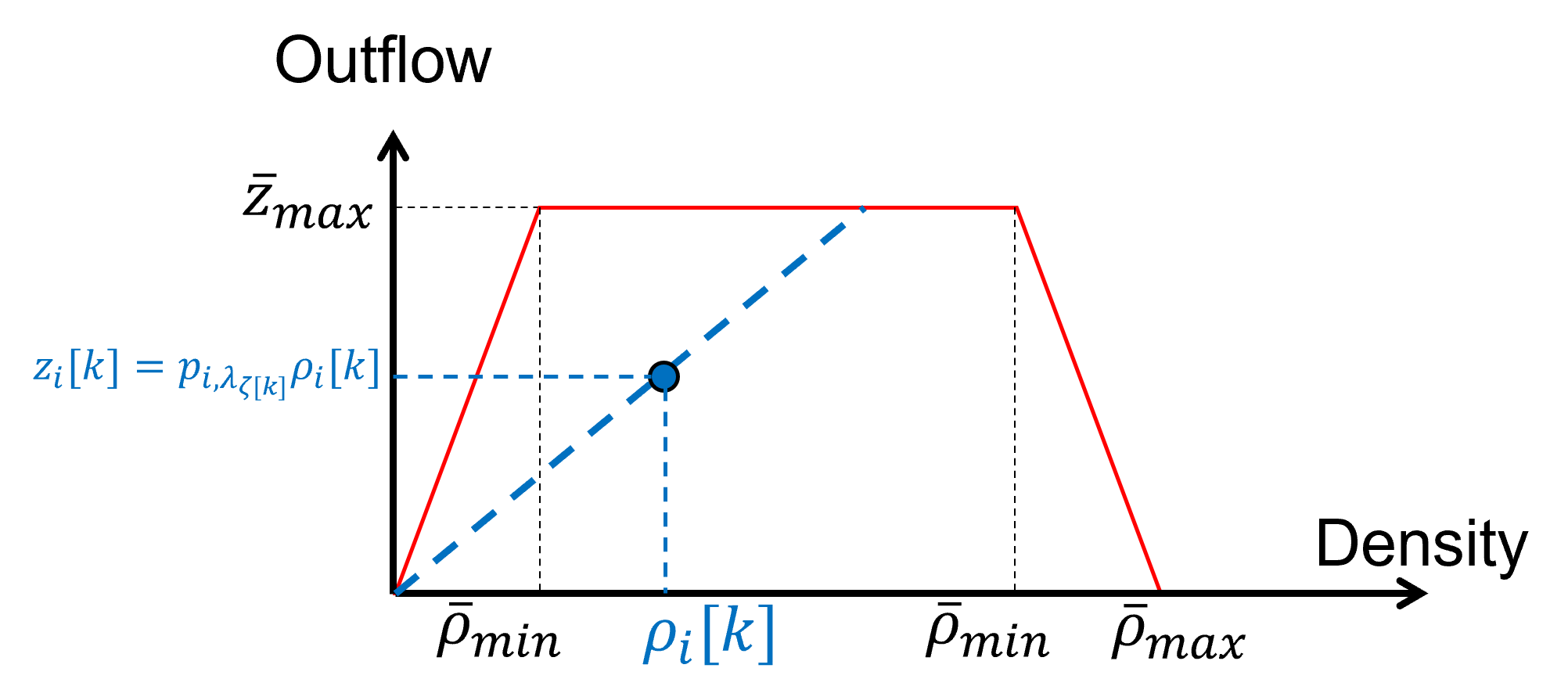}
\caption{Graphic representation of constraint equation \eqref{fnd_all} imposed by the trapezoid fundamental diagram\cite{WU2011empirical}.
The left limit of the diagram corresponds to constraint (\ref{fnd2}); the top limit to constraint (\ref{fnd3}); and the right limit to constraint (\ref{fnd4}).
}
\label{Fundamentaldiagram}
\end{figure}



The objective of this paper is to determine boundary inflow  $u_i[k]$ at every road $i\in \mathcal{V}$ and every discrete time $k$ such that the traffic coordination cost defined by
\begin{equation}\label{cost function}
    \mathrm{J}=\sum_{j=1}^{n_c}\left(\sum_{i\in\mathcal{V}_{in}}u_i^2[k+j]+\beta\sum_{i\in\mathcal{V}}\rho_i^2[k+j]\right)
\end{equation}
is minimized, where scaling parameter $\beta>0$ is constant.




\section{Problem Specification}\label{Problem Specification}




We can express the requirements from Section~\ref{sec:problem} in Linear Temporal Logic (LTL).
Every LTL formula consists of a set of atomic propositions, logical operators, and temporal operators. Logical operators include $\lnot$ (``negation''), $\vee$ (``disjunction''), $\wedge$ (``conjunction''), and $\Rightarrow$ (``implication''). LTL formulae also use temporal operators  $\square$ (``always''), $\bigcirc$ (``next''), $\lozenge$ (``eventually''), and $\mathcal{U}$ (``until'').

The traffic evolution governed by \eqref{eq:update} 
must satisfy the feasibility requirements (equations (\ref{C1})-(\ref{C3})), leading to the requirements:
\begin{subequations}
\begin{equation}
    \bigwedge_{i\in \mathcal{V}_{in}}\square\left(u_i\geq 0\right),
\end{equation}
\begin{equation}
    \square\left(
    \sum_{i\in \mathcal{V}_{in}}u_i=u_0\right),
\end{equation}
\begin{equation}
    \bigwedge_{i\in \mathcal{V}}\square\left(\rho_i\geq 0\right),
\end{equation}
\begin{equation}
    \bigwedge_{i\in \mathcal{V}}\square\left(\rho_i\leq \bar{\rho}_{\max} \right).
\end{equation}
\end{subequations}
Additionally the movement phase rotation can be expressed as:
\begin{equation}
        \square\left(\left(\lambda_\zeta,\lambda_\gamma\right)\in \mathcal{M}\right).
\end{equation}

We can also concisely express the FD constraints (Eq. \eqref{zetagamma} and Eqs. (\ref{fnd1})-(\ref{fnd4})), leading to the LTL requirements:
\begin{subequations}
\begin{equation}
\bigwedge_{i\in \mathcal{V}}\square\left(z_i\geq 0\right),
\end{equation}
\begin{equation}
    \bigwedge_{i\in \mathcal{V}}\square\left(z_i\leq {\bar{z}_{\max}\rho_i\over \bar{\rho}_{\min}}\right),
\end{equation}
\begin{equation}
    \bigwedge_{i\in \mathcal{V}}\square\left(z_i\leq \bar{z}_{\max}\right),
\end{equation}
\begin{equation}
    \bigwedge_{i\in \mathcal{V}}\square\left(z_i\leq \frac{\bar{z}_{\max}\left(\rho_i-\bar{\rho}_{\max}\right)} {(\bar{\rho}_{\rm{mid}}-\bar{\rho}_{\max})}\right).
\end{equation}
\end{subequations}
The objective of traffic congestion control is to satisfy the following  liveness conditions:
\begin{equation}\label{Liveness}
   \Diamond\left( \left|\bigwedge_{i\in \mathcal{V}_{out}}z_i-u_0\right|<\epsilon\right),
\end{equation}
where $\epsilon$ is constant and obtained in Section \ref{Traffic Network Dynamics}. Liveness condition \eqref{Liveness} specifies the reachability of the traffic state to the steady-state condition where the network inflow and outflow are the same. Theorem \ref{LivenesTheorem} presented in Section \ref{Traffic Network Dynamics} proves that the liveness condition \eqref{Liveness} is satisfied if the proposed first order dynamics is used to model the traffic coordination.  

%

\section{Traffic Network Dynamics}
\label{Traffic Network Dynamics}
To obtain the traffic network dynamics, we define tendency probability matrix $\mathbf{Q}\left(\zeta[k]\right)=\left[q_{i,j,\zeta[k]}\right]\in \mathbb{R}^{N\times N}$, outflow probability matrix
\begin{equation}
    \mathbf{P}\left(\zeta[k]\right)=\mathbf{diag}\left(p_{1,\lambda_{\zeta[k]}},\cdots,p_{N,\lambda_{\zeta[k]}}\right)\in \mathbb{R}^{N\times N},\qquad \lambda\in \mathcal{L},
\end{equation}
and 
\begin{equation}
    \mathbf{A}\left(\zeta[k]\right)=\mathbf{I}+\left(\mathbf{Q}\left(\zeta[k]\right)-\mathbf{I}\right)\mathbf{P}\left(\zeta[k]\right),\qquad \zeta \in \sigma,~k\in \mathbb{N}.
\end{equation}

We also define matrix $\mathbf{B}=\left[b_{ij}\right]\in \mathbb{R}^{N\times N}$ with the $(i,j)$ entry that is defined as follows:
\begin{equation}\label{matrixB}
    b_{ij}=\begin{cases}
    1&i\in \mathcal{I}_{j}\\
    0&\mathrm{otherwise}
    \end{cases}
    .
\end{equation}
By defining the state vector $\mathbf{x}[k]=\begin{bmatrix}\rho_1[k]&\cdots&\rho_N[k]\end{bmatrix}^\mathsf{T}$ and input vector $\mathbf{u}[k]=\begin{bmatrix}u_1[k]&\cdots&u_{N_{in}}[k]\end{bmatrix}^\mathsf{T}$, and imposing the CTM given in \eqref{eq:update}, the traffic network dynamics is obtained as follows:
\begin{equation}\label{singlestepdyn}
    \mathbf{x}[k+1]=\mathbf{A}\left(\zeta[k]\right)\mathbf{x}[k]+\mathbf{B}\mathbf{u}[k],\qquad \zeta=\sigma.
\end{equation}
Given the above definitions, matrices $\mathbf{P}\left(\zeta[k]\right)$ and $\mathbf{Q}\left(\zeta[k]\right)$ hold the following properties:
\begin{property}\label{property3}
Entries of the diagonal matrix $\mathbf{P}\left(\zeta[k]\right)$ are positive and not greater than $1$ because $p_{i,\lambda_\zeta[k]}\in \left(0,1\right]$ for every $i\in\mathcal{V}$. 
\end{property}

\begin{property}\label{property4}
Matrix $\mathbf{Q}\left(\zeta[k]\right)$ is a non-negative matrix because  $q_{j,i,\lambda_\zeta[k]}\in \left(0,1\right]$ for every $i\in\mathcal{V}$.
\end{property}

\begin{property}\label{property5}
Diagonal entries of matrix $\mathbf{Q}\left(\zeta[k]\right)$ are $0$ because  $i\notin \mathcal{O}_i$ for every $i\in\mathcal{V}$.
\end{property}

\begin{property}\label{property6}
Because road elements in the NOIR are unidirectional interconnected, graph $\mathcal{G}\left(\mathcal{V},\mathcal{E}\right)$ holds the property that:
\[
\left(i,j\right)\in \mathcal{E}\implies \left(j,i\right)\notin\mathcal{E},\quad \forall i,j\in \mathcal{V},
\]
which indicates that off-diagonal entries of matrix $\mathbf{Q}\left(\zeta[k]\right)$ follow the relation that:
\[
Q_{i,j,\lambda_{\zeta[k]}}\neq0\implies Q_{j,i,\lambda_{\zeta[k]}}=0,\quad \forall i\neq j\cap i,j\in \mathcal{V}.
\]
\end{property}

\begin{property}\label{property7}
At each discrete time $k\in \mathbb{N}$,
\begin{subequations}
\begin{equation}
     \sum_{j=1}^{N}Q_{j,i,\lambda_\zeta[k]}=0, \quad \forall i\in\mathcal{V}_{out}
\end{equation}
\begin{equation}
    \sum_{j\in\mathcal{O}_i} q_{j,i,\lambda_\zeta[k]} = \sum_{j=1}^{N}Q_{j,i,\lambda_\zeta[k]}=1, \quad \forall i\in\mathcal{V}\setminus\mathcal{V}_{out}
\end{equation}
\end{subequations}
\end{property}
\begin{theorem}
If Properties \ref{property3}-\ref{property7} are all satisfied at each discrete time $k$, the traffic dynamics \eqref{singlestepdyn} is BIBO stable.
\end{theorem}
\textbf{Proof:} According to the Gershgorin circle theorem\cite{horn2012matrix}, every eigenvalue of matrix $\mathbf{Q}\left(\zeta[k]\right)-\mathbf{I}$ lies within at least one of the Gershgorin discs $D(-1,1)$. Because entries of matrix $\mathbf{P}\left(\zeta[k]\right)$ are all in the interval $(0,1]$, eigenvalues of matrix $\mathbf{A}\left(\zeta[k]\right)$ must be located within the discs $D(0,1)$\cite{boyd2018Algebra}. If this is not satisfied and some of the eigenvalues of $\mathbf{A}\left(\zeta[k]\right)$  are $1$, then, $0$ is an element of the spectrum of matrix $\mathbf{Q}\left(\zeta[k]\right)-\mathbf{I}$, which indicates that the matrix $\mathbf{Q}\left(\zeta[k]\right)-\mathbf{I}$ is not full rank, i.e. $\mathrm{rank}(\mathbf{Q}\left(\zeta[k]\right)-\mathbf{I})<N$. However, considering the properties \ref{property5} and \ref{property6} of matrix $\mathbf{Q}$, it can be seen that rows of matrix $\mathbf{Q}\left(\zeta[k]\right)-\mathbf{I}$ are independent, which implies that the rank of matrix $\mathrm{rank}(\mathbf{Q}\left(\zeta[k]\right)-\mathbf{I})=N$. Therefore, since the assumption of matrix $\mathbf{A}\left(\zeta[k]\right)$ can not be satisfied, we can draw the conclusion that eigenvalues of matrix $\mathbf{A}$ must located within the discs $D(0,1)$ strictly, i.e. the spectral radius $\rho(\mathbf{A}\left(\zeta[k]\right))<1$. Then, since eigenvalues of matrix $\mathbf{A}\left(\zeta[k]\right)$ are within the unit circle strictly at each discrete time $k$, the traffic dynamics \eqref{singlestepdyn} is BIBO stable\cite{gu2012discrete,Liu2021ConservationBased}.

\begin{theorem}\label{LivenesTheorem}
If Properties \ref{property3}-\ref{property7} of matrices $\mathbf{P}\left(\zeta[k]\right)$ and $\mathbf{Q}\left(\zeta[k]\right)$ are all satisfied, then, liveness condition \eqref{Liveness} is satisfied.
\end{theorem}
\begin{proof}
The traffic dynamics \eqref{singlestepdyn} can be rewritten as
\begin{equation}\label{singlestepdyn2}
    \mathbf{x}[k+1]=\mathbf{x}[k]+\left(\mathbf{Q}\left(\zeta[k]\right)-\mathbf{I}\right)\mathbf{z}[k]+\mathbf{B}\mathbf{u}[k],\qquad \zeta=\sigma,~k\in \mathbb{N}.
\end{equation}
Therefore,
\begin{equation}\label{singlestepdyn22}
    \mathbf{1}_{1\times N}\left(\mathbf{x}[k+1]-\mathbf{x}[k]\right)=\mathbf{1}_{1\times N}\left(\mathbf{Q}\left(\zeta[k]\right)-\mathbf{I}\right)\mathbf{z}[k]+u_0,~ \zeta=\sigma.
\end{equation}
where $\mathbf{1}_{1\times N}\in\mathbb{R}^{1\times N}$ is a row vector, with entries that are all $1$, $u_0=\mathbf{1}_{1\times N}\mathbf{B}\mathbf{u}[k]$ is constant at every discrete time $k$ per condition \eqref{C2}. Because traffic dynamics \eqref{singlestepdyn} is BIBO stable, there exists a discrete time $k_s$ such that
\begin{subequations}
\begin{equation}
    \left|\mathbf{1}_{1\times N}\left(\mathbf{x}[k+1]-\mathbf{x}[k]\right)\right|<\delta_1,\qquad \forall k\geq k_s,
\end{equation}
\begin{equation}
    \left|\mathbf{1}_{1\times N}\left(\mathbf{Q}\left(\zeta[k]\right)-\mathbf{I}\right)\mathbf{z}[k]-\sum_{i\in \mathcal{V}_{out}}z_i[k]\right|<\delta_2,\qquad \forall k\geq k_s.
\end{equation}
\end{subequations}
Note the $z_i[k]>0$ at every time $k$, therefore,  
\begin{equation}\label{Liveness}
  \left|\bigwedge_{i\in \mathcal{V}_{out}}z_i[k]-u_0\right|<\epsilon =\delta_1+\delta_2\qquad \forall k\geq k_s.
\end{equation}
\end{proof}
\section{Traffic Congestion Control}\label{Traffic Congestion Control}
We use MPC to determine the boundary control $\mathbf{u}[k]$ at every discrete time $k$ by solving a quadratic programming problem with a quadratic cost  and linear constraints imposing the feasibility conditions into management of traffic coordination. To this end, we first define matrix multiplication process
\begin{equation}
    \mathbf{H}\left(\zeta[k+i]\right)=\mathbf{H}\left(\zeta[k+i-1]\right)\mathbf{A}\left(\zeta[k]\right),\qquad i\in \sigma.
\end{equation}
subject to 
\begin{equation}
    \mathbf{H}\left(\zeta[k]\right)=\mathbf{I}.
\end{equation}
We then apply \eqref{singlestepdyn} to predict the traffic evolution within the next $n_c$ sampling times by 
\begin{equation}\label{iterative model}
    \mathbf{X}[k]=\mathbf{G}_1\left(\zeta[k]\right)\mathbf{x}[k]+\mathbf{G}_2\left(\zeta[k]\right)\mathbf{U}[k],\qquad k\in \mathbb{N},~\zeta\in \sigma
\end{equation}
where $\mathbf{X}[k]=\begin{bmatrix}
    \mathbf{x}^\mathsf{T}[k+1]&\cdots&\mathbf{x}^\mathsf{T}[k+n_c]\end{bmatrix}^\mathsf{T}
    $, $\mathbf{U}[k]=\begin{bmatrix}
    \mathbf{u}^\mathsf{T}[k]&\cdots&\mathbf{x}^\mathsf{T}[k+n_c-1]\end{bmatrix}^\mathsf{T}$, and matrices $\mathbf{G}_1\left(\zeta[k]\right)$ and $\mathbf{G}_2\left(\zeta[k]\right)$ are defined by \eqref{MatrixG}  and \eqref{matrixH}, respectively. 
\begin{strip}
    \begin{subequations}\label{state iternative model}
\begin{equation}\label{MatrixG}
    \mathbf{G}_1\left(\zeta[k]\right)=\begin{bmatrix}
    \mathbf{H}\left(\zeta[k+1]\right)\\
    \vdots\\
    \mathbf{H}\left(\zeta[k+n_c]\right)
    \end{bmatrix}
    \in \mathbb{R}^{Nn_c\times N},\qquad k\in \mathbb{N},~\zeta \in \sigma
    ,
\end{equation}
\begin{equation}\label{matrixH}
    \mathbf{G}_2\left(\zeta[k]\right)=\begin{bmatrix}
    \mathbf{H}\left(\zeta[k]\right) & \mathbf{0} & \mathbf{0} &\cdots & \mathbf{0}\\
    \mathbf{H}\left(\zeta[k+1]\right) & \mathbf{H}\left(\zeta[k]\right) & \mathbf{0} & \cdots & \mathbf{0}\\
    \mathbf{H}\left(\zeta[k+2]\right) & \mathbf{H}\left(\zeta[k+1]\right) & \mathbf{H}\left(\zeta[k]\right) & \cdots & \mathbf{0}\\
    \vdots & \vdots & \vdots & &\vdots \\
    \mathbf{H}\left(\zeta[k+n_c-1]\right) & \mathbf{H}\left(\zeta[k+n_c-2]\right)& \mathbf{H}\left(\zeta[k+n_c-3]\right) & \cdots & \mathbf{H}\left(\zeta[k]\right)
    \end{bmatrix}
    \left(\mathbf{B}\otimes \mathbf{1}_{1\times n_c}\right)\in \mathbb{R}^{Nn_c\times Nn_c},\qquad k\in \mathbb{N},~\zeta \in \sigma
\end{equation}
\end{subequations}
\end{strip}
In Eq. \eqref{matrixH} $\otimes$ is the Kronecker product symbol and $\mathbf{1}_{1\times n_c}\in \mathbb{R}^{1\times n_c}$ is a row vector with the entries that are all $1$.

The cost function $\mathrm{J}$, previously defined in  \eqref{cost function}, can be rewritten as follows:
\begin{equation}\label{cost function state}
\begin{split}
\mathrm{J}\left(\mathbf{U}[k],\zeta[k]\right)=&
{1\over 2}\mathbf{U}^\mathsf{T}[k]\mathbf{W}_1\left(\zeta[k]\right)\mathbf{U}[k]+\mathbf{W}_2^\mathsf{T}\left(\zeta[k]\right)\mathbf{U}[k]\\+&\mathbf{W}_3\left(\zeta[k]\right),
\end{split}
\end{equation}
where 
\begin{subequations}
    \begin{equation}
    \mathbf{W}_1\left(\zeta[k]\right)=\mathbf{I}+\beta\mathbf{G}_2^\mathsf{T}\left(\zeta[k]\right)\mathbf{G}_2\left(\zeta[k]\right),
    \end{equation}
    \begin{equation}
    \mathbf{W}_2^\mathsf{T}\left(\zeta[k]\right)=\beta\mathbf{x}^\mathsf{T}[k]\mathbf{G}_1^\mathsf{T}\left(\zeta[k]\right)\mathbf{G}_2\left(\zeta[k]\right),
    \end{equation}
    \begin{equation}
    \mathbf{W}_3\left(\zeta[k]\right)={1\over 2}\beta\mathbf{x}^\mathsf{T}[k]\mathbf{G}_1^\mathsf{T}\left(\zeta[k]\right)\mathbf{G}_1\left(\zeta[k]\right)\mathbf{x}[k].
    \end{equation}
\end{subequations}
Note that $\mathbf{W}_3\left(\zeta[k]\right)$ can be removed from cost function  \eqref{cost function state} because it does not depend on $\mathbf{U}[k]$ at every discrete time $k$. Therefore, 
\begin{equation}\label{cost function final}
    \mathrm{J}'={1\over 2}\mathbf{U}^\mathsf{T}[k]\mathbf{W}_1\left(\zeta[k]\right)\mathbf{U}[k]+\mathbf{W}_2^\mathsf{T}\left(\zeta[k]\right)\mathbf{U}[k]
\end{equation}
is considered as the cost function of traffic coordination, and the optimal control variable
\begin{equation}
    \mathbf{u}^*[k]=\begin{bmatrix}
    \mathbf{I}_{N_{in}}&\mathbf{0}_{N_{in}\times N_{in}\left(N_c-1\right)}
    \end{bmatrix}
    \mathbf{U}^*[k]
\end{equation}
is assigned by determining $\mathbf{U}^*[k]$ by solving of the following optimization problem:
\[
\min \mathrm{J}'=\min~\left({1\over 2}\mathbf{U}^\mathsf{T}[k]\mathbf{W}_1\left(\zeta[k]\right)\mathbf{U}[k]+\mathbf{W}_2^\mathsf{T}\left(\zeta[k]\right)\mathbf{U}[k]\right)
 \]
subject to
\begin{subequations}\label{feasibility constrains}
\begin{equation}\label{iterative model}
    -\mathbf{G}_2\left(\zeta[k]\right)\mathbf{U}[k]\leq \mathbf{G}_1\left(\zeta[k]\right)\mathbf{x}[k],
\end{equation}
\begin{equation}\label{iterative model}
    \mathbf{G}_2\left(\zeta[k]\right)\mathbf{U}[k]\leq -\mathbf{G}_1\left(\zeta[k]\right)\mathbf{x}[k]+\bar{\rho}_{max}\mathbf{1}_{Nn_c\times 1},
\end{equation}
\begin{equation}\label{inequality condition}
     \left(\mathbf{W}_4\left(\zeta[k]\right)-{\bar{z}_{max}\over \bar{\rho}_{min}}\mathbf{I}\right)\mathbf{G}_2\left(\zeta[k]\right)\mathbf{U}[k]\leq -\mathbf{G}_1\left(\zeta[k]\right)\mathbf{x}[k],
\end{equation}
\begin{equation}\label{inequality condition}
     \mathbf{W}_4\left(\zeta[k]\right)\mathbf{G}_2\left(\zeta[k]\right)\mathbf{U}[k]\leq -\mathbf{G}_1\left(\zeta[k]\right)\mathbf{x}[k]+\mathbf{v}_1,
\end{equation}
\begin{equation}\label{inequality condition}
\begin{split}
    \left(\mathbf{W}_4\left(\zeta[k]\right)+{\bar{z}_{max}\over \bar{\rho}_{max}-\bar{\rho}_{mid}}\mathbf{I}\right)\left(\zeta[k]\right)\mathbf{G}_2\left(\zeta[k]\right)\mathbf{U}[k]&\leq\\ -\mathbf{G}_1\left(\zeta[k]\right)\mathbf{x}[k]+\mathbf{v}_2&,
\end{split}
\end{equation}
\begin{equation}\label{equality condition equation}
    \mathbf{I}_{n_c}\otimes\mathbf{1}_{1\times N_{in}}\mathbf{U}[k]=-u_0\mathbf{1}_{n_c\times 1},
\end{equation}
\end{subequations}
where
\begin{subequations}
\begin{equation}
\mathbf{W}_4\left(\zeta[k]\right)=
    \begin{bmatrix}
    \mathbf{P}\left(\zeta[k+1]\right)&\cdots&\mathbf{0}\\
    \vdots&\ddots&\vdots\\
    \mathbf{0}&\cdots&\mathbf{P}\left(\zeta[k+n_c]\right)
    \end{bmatrix}
    \in \mathbb{R}^{Nn_c\times Nn_c},
\end{equation}
\begin{equation}
\mathbf{v}_1=\bar{z}_{max}\mathbf{1}_{Nn_c},
\end{equation}
\begin{equation}
\mathbf{v}_2={\bar{z}_{max}\bar{\rho}_{max}\over \bar{\rho}_{max}-\bar{\rho}_{mid}}\mathbf{1}_{Nn_c}.
\end{equation}
 \end{subequations}

\begin{table}
\begin{adjustbox}{max width=.5\textwidth,center}
    \centering
\begin{tabular}{|M{0.18cm}|c|M{0.9cm}|M{0.18cm}|c|M{0.9cm}|}
  \hline
    ID & Name & Direction & ID & Name & Direction  \\
   \hline
   1 & N10th St.(E McKinley St.- E Pierce St.)& S & 2 &N11th St.(E McKinley St.- E Pierce St.)& S \\
    \hline
    3 & E Fillmore St.(N12th St.- N13th St.)& W & 4 &E Pierce St.(N12th St.- N13th St.)& W \\
    \hline
    5 & E Taylor St.(N12th St.- N13th St.)& W & 6 &E Fillmore St.(N7th St.- N9th St.)& E \\
    \hline
    7 & N9th St.(E Taylor St.- E Polk St.)& N & 8 &N11th St.(E Taylor St.- E Polk St.)& N \\
    \hline
    9 & N12th St.(E McKinley St.- E Pierce St.)& S & 10 &E Pierce St.(N7th St.- N9th St.)& E \\
    \hline
    11 & N10th St.(E Taylor St.- E Polk St.)& N & 12 &E Fillmore St.(N7th St.- N9th St.)& W \\
    \hline
    13 & N9th St.(E Taylor St.- E Polk St.)& S & 14 &N11th St.(E Taylor St.- E Polk St.)& S \\
    \hline
    15 &N12th St.(E McKinley St.- E Pierce St.)& N & 16 &E Pierce St.(N7th St.- N9th St.)& W \\
    \hline
    17 & N10th St.(E Taylor St.- E Polk St.)& S & 18 &N10th St.(E McKinley St.- E Pierce St.)& N \\
    \hline
    19 & N11th St.(E McKinley St.- E Pierce St.)& N & 20 &E Fillmore St.(N12th St.- N13th St.)& E \\
    \hline
    21 & E Pierce St.(N12th St.- N13th St.)& E &22 &E Taylor St.(N12th St.- N13th St.)& E \\
    \hline
    23 & N9th St.(E Fillmore St.- E Taylor St.)& S & 24 &N9th St.(E Fillmore St.- E Taylor St.)& S \\
    \hline
    25 & N9th St.( E Pierce St.- E Fillmore St.)& S & 26 &N11th St.(E Fillmore St.- E Taylor St.)& S \\
    \hline
    27 & N11th St.( E Pierce St.- E Fillmore St.)& S & 28 &N12th St.(E Pierce St.- E Fillmore St.)& N \\
    \hline
    29 & N12th St.(E Fillmore St.- E Taylor St.)& N & 30 &E Fillmore St.(N9th St.- N10th St.)& W \\
    \hline
    31 &E Fillmore St.(N10th St.- N11th St.)& W & 32 &E Fillmore St.(N11th St.- N12th St.)& W \\
    \hline
    33 &E Pierce St.(N9th St.- N10th St.)& W & 34 &E Pierce St.(N10th St.- N11th St.)& W \\
    \hline
    35 & E Pierce St.(N10th St.- N11th St.)& W & 36 &E Pierce St.(N11th St.- N12th St.)& W \\
    \hline
    37 & N10th St.(E Fillmore St.- E Taylor St.)& S & 38 &N10th St.(E Pierce St.- E Fillmore St.)& S \\
    \hline
    39 &  E Taylor St.(N9th St.- N10th St.)& W & 40 &E Taylor St.(N10th St.- N11th St.)& W \\
    \hline
    41 & E Taylor St.(N11th St.- N12th St.)& W & 42 &N9th St.(E Fillmore St.- E Taylor St.)& N \\
    \hline
    43 & N9th St.(E Fillmore St.- E Taylor St.)& N & 44 &N9th St.( E Pierce St.- E Fillmore St.)& N \\
    \hline
    45 &N11th St.(E Fillmore St.- E Taylor St.)& N & 46 &N11th St.( E Pierce St.- E Fillmore St.)& N \\
    \hline
    47 & N12th St.(E Pierce St.- E Fillmore St.)& S & 48 &N12th St.(E Fillmore St.- E Taylor St.)& S \\
    \hline
    49 &E Fillmore St.(N9th St.- N10th St.)& E & 50 &E Fillmore St.(N10th St.- N11th St.)& E \\
    \hline
    51 & E Fillmore St.(N11th St.- N12th St.)& E &52 &E Pierce St.(N9th St.- N10th St.)& E \\
    \hline
    53 & E Pierce St.(N10th St.- N11th St.)& E & 54 &E Pierce St.(N10th St.- N11th St.)& E \\
    \hline
    55 & E Pierce St.(N11th St.- N12th St.)& E & 56 &N10th St.(E Fillmore St.- E Taylor St.)& N \\
    \hline
    57 & N10th St.(E Pierce St.- E Fillmore St.)& N & 58 &E Taylor St.(N9th St.- N10th St.)& E \\
    \hline
    59 & E Taylor St.(N10th St.- N11th St.)& E & 60 &E Taylor St.(N11th St.- N12th St.)& E \\
    \hline
\end{tabular}
\end{adjustbox}
\caption{Road elements of the example NOIR of Phoenix City}
\label{tab:road_of_phoenix_NOIR}
\end{table}

\begin{table}[]
\caption{Number of movement phases at every junction $i\in \mathcal{W}$}
\begin{adjustbox}{max width=.5\textwidth,center}
    \centering
    \begin{tabular}{|c|c|c|c|c|c|c|c|c|c|c|c|c|c|}
    \hline
   \multicolumn{14}{|c|}{Number of movement phases at junction $i$ denoted by $r_i$}\\
   \hline
    $r_1$&$r_2$&$r_3$&$r_4$&$r_5$&$r_6$&$r_7$&$r_8$&$r_9$&$r_{10}$&$r_{11}$&$r_{12}$&$r_{13}$&$r_{14}$\\
    \hline
    3&3&4&4&3&4&4&3&3&3&4&4&3&4\\
    \hline
    \end{tabular}
    \label{tab:my_label2}
\end{adjustbox}
\end{table}

\begin{table}[]
\caption{Active incoming roads at junctions defined by $\mathcal{W}$ over the entire cycle of length $n_c=12$}
\begin{adjustbox}{max width=.5\textwidth,center}
    \centering
    \begin{tabular}{|c|c|c|c|c|c|c|c|c|c|c|c|c|}
    \hline
    &\multicolumn{12}{|c|}{Active movement phases over a cycle}\\
    $i\in \mathcal{W}$&$k+1$&$k+2$&$k+3$&$k+4$&$k+5$&$k+6$&$k+7$&$k+8$&$k+9$&$k+10$&$k+11$&$k+12$\\
    \hline
    1&     1&    34   & 52&     1   & 34&    52   &  1&    34   & 52&     1   & 34&    52\\
    \hline
    2&    48 &    5  &  60 &   48  &   5 &   60  &  48 &    5  &  60 &   48  &   5 &   60\\
       \hline
    3&    47  &   3 &   29  &  51 &   47  &   3 &   29  &  51 &   47  &   3 &   29  &  51\\
      \hline
    4&     9   &  4&    28   & 55&     9   &  4&    28   & 55&     9   &  4&    28   & 55\\
 \hline
    5&    25    &30    &43    &25    &30    &43    &25    &30    &43    &25    &30    &43\\
 \hline
    6&    38&    31   & 56&    49   & 38&    31   & 56&    49   & 38&    31   & 56&    49\\
            \hline
    7&    27 &   32  &  45 &   50  &  27 &   32  &  45 &   50  &  27 &   32  &  45 &   50\\
            \hline
    8&    24  &  42 &    6  &  24 &   42  &   6 &   24  &  42 &    6  &  24 &   42  &   6\\
            \hline
    9&    23   & 39&     7   & 23&    39   &  7&    23   & 39&     7   & 23&    39   &  7\\
            \hline
    10&    33    &44    &10    &33    &44    &10    &33    &44    &10    &33    &44    &10\\
            \hline
    11&    26&    41   &  8&    59   & 26&    41   &  8&    59   & 26&    41   &  8&    59\\
            \hline
    12&     2 &   36  &  46 &   54  &   2 &   36  &  46 &   54  &   2 &   36  &  46 &   54\\
             \hline
    13&    35  &  57 &   53  &  35 &   57  &  53 &   35  &  57 &   53  &  35 &   57  &  53\\
            \hline
    14&    37   & 40&    11   & 58&    37   & 40&    11   & 58&    37   & 40&    11   & 58\\
            \hline
    \end{tabular}
    \label{tab:my_label}
    \end{adjustbox}
\end{table}

\begin{table}[]
\caption{Simulation parameters for every  $i\in\mathcal{V}$}
\begin{adjustbox}{max width=.5\textwidth,center}
    \centering
    \begin{tabular}{|c|c|c|c|c|}
    \hline
    $\bar{z}_{i,max}$&$\bar{\rho}_{i,min}$&$\bar{\rho}_{i,mid}$&$\bar{\rho}_{i,max}$&$u_0$\\
    \hline
    $20$&$20$&$40$&$55$&$50$\\
    \hline
    \end{tabular}
    \label{paramerte}
\end{adjustbox}
\end{table}

\section{Simulation Results}\label{Simulation Results}
We simulate traffic congestion control in a selected area in Downtown Phoenix with the map and NOIR shown in Fig. \ref{realmap}. The NOIR consisting of $60$ unidirectional roads with the identification numbers defined by set $\mathcal{V}=\left\{1,\cdots,60\right\}$ and the names listed in Table \ref{tab:road_of_phoenix_NOIR}.  Interconnections between roads are defined by   $\mathcal{G}\left(\mathcal{V},\mathcal{E}\right)$ with  node set $\mathcal{V}$ and edge set $\mathcal{E}$, where $\mathcal{V}=\mathcal{V}_{in}\bigcup \mathcal{V}_{out}\bigcup \mathcal{V}_I$, $\mathcal{V}_{in}=\left\{1,\cdots,11\right\}$, $\mathcal{V}_I=\left\{12,\cdots,22\right\}$. As shown in  Fig. \ref{realmap},  the NOIR consists of $14$ junctions defined by set $\mathcal{W}=\left\{1,\cdots,14\right\}$. Without loss of generality, for definition of movement phases, we make the following assumption in addition to Assumptions \ref{repeatedmovephase} and \ref{concurrent}:
\begin{assumption}\label{assump3}
   At every discrete time $k\in \mathbb{N}$, traffic can enter junction $i\in \mathcal{W}$ through a single incoming road which is called active incoming road.    
\end{assumption}
By imposing Assumption \ref{assump3}, the number of movement phases is either $3$ or $4$ at every junction $i\in \mathcal{W}$ (see Table \ref{tab:my_label2}). Therefore, the NOIR cycle is completed in $n_c=12$ time steps. Table  \ref{tab:my_label} lists the active \underline{incoming roads} over the NOIR cycle of the length $n_c=12$.
\begin{figure}[h]
\centering
\subfigure[]{\includegraphics[width=.98\linewidth]{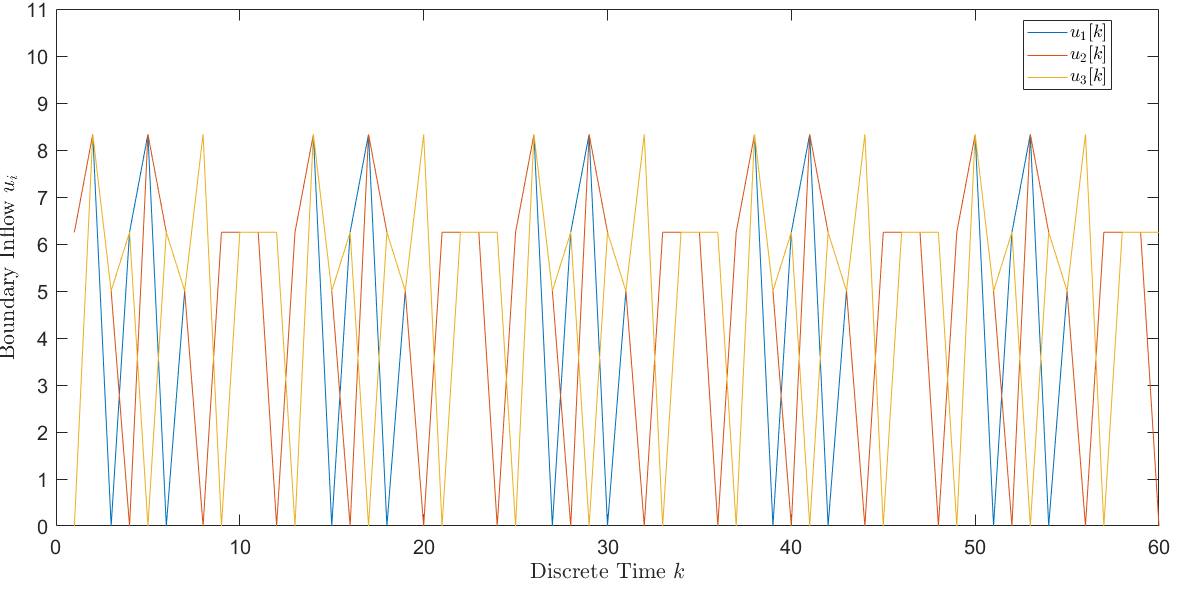}}
\subfigure[]{\includegraphics[width=.98\linewidth]{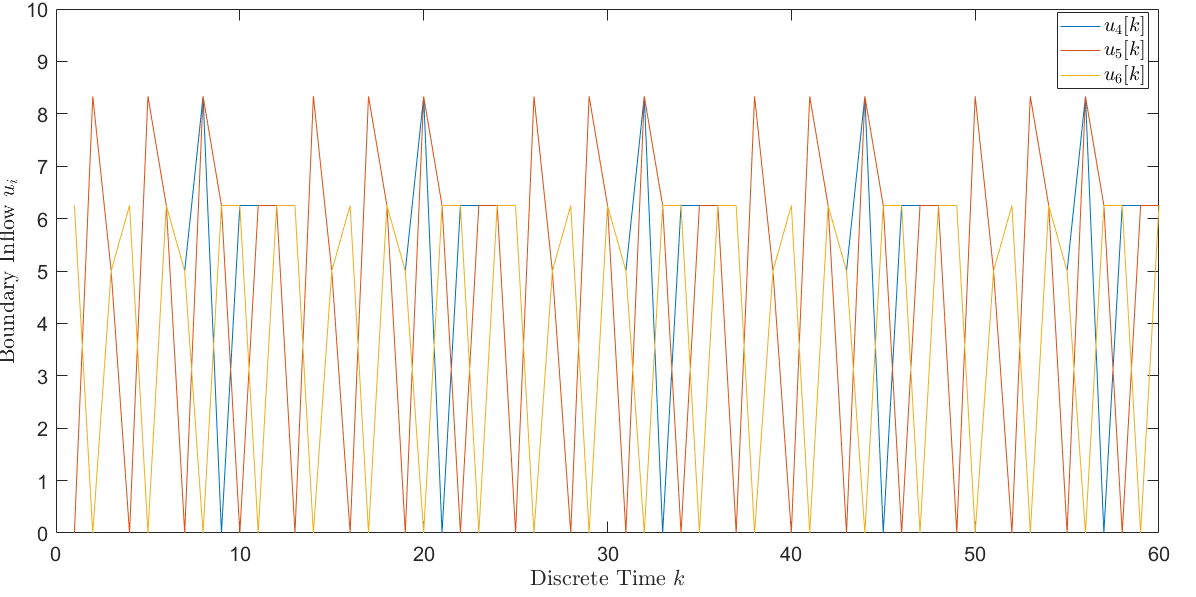}}
\subfigure[]{\includegraphics[width=.98\linewidth]{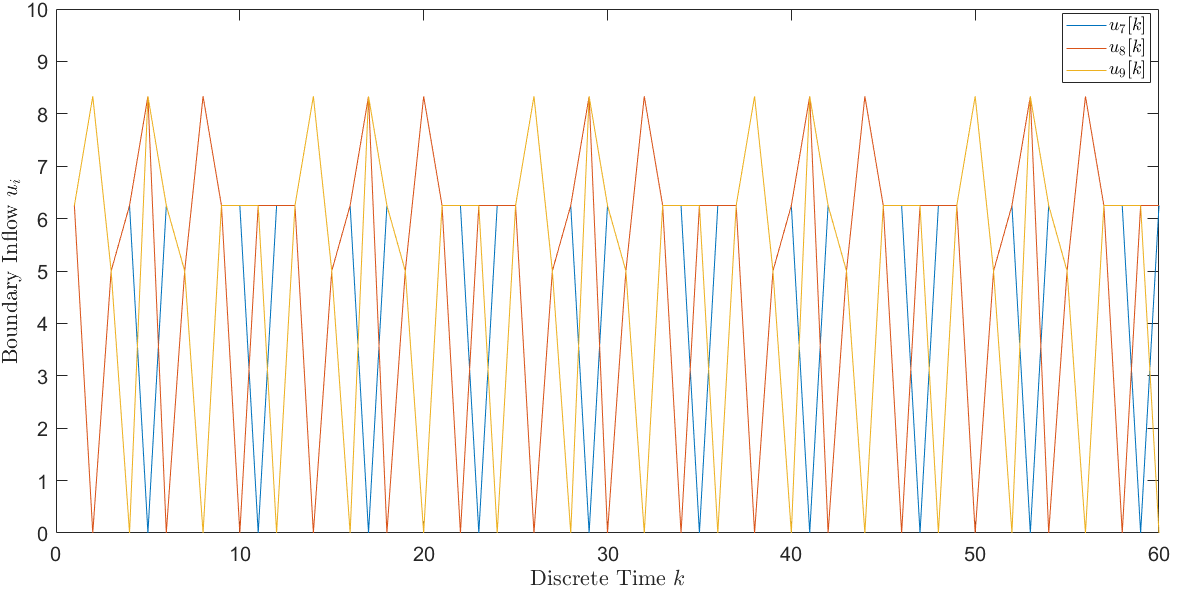}}
\subfigure[]{\includegraphics[width=.98\linewidth]{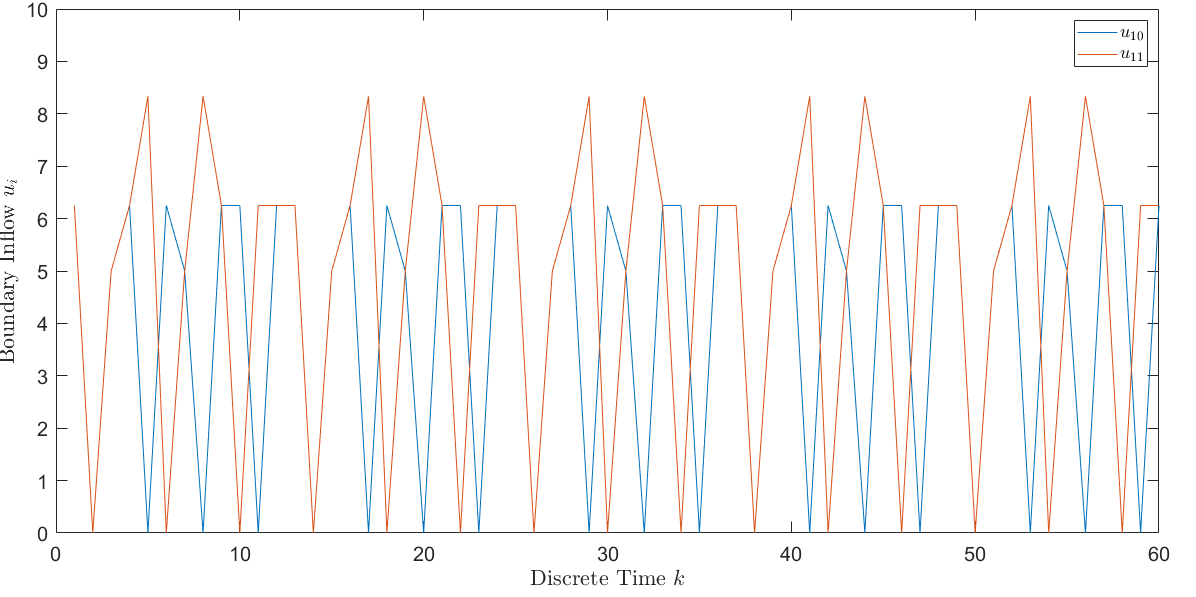}}
\caption{The optimal boundary inflow at inlet roads (a) $1,2,3\in \mathcal{V}_{in}$, (b) $4,5,6\in \mathcal{V}_{in}$, (c) $7,8,9\in \mathcal{V}_{in}$, and (d) $10,11\in \mathcal{V}_{in}$. }
\label{inflow}
\end{figure}

\begin{figure*}[h]
\centering
\subfigure[$k=15$]{\includegraphics[width=.33\linewidth]{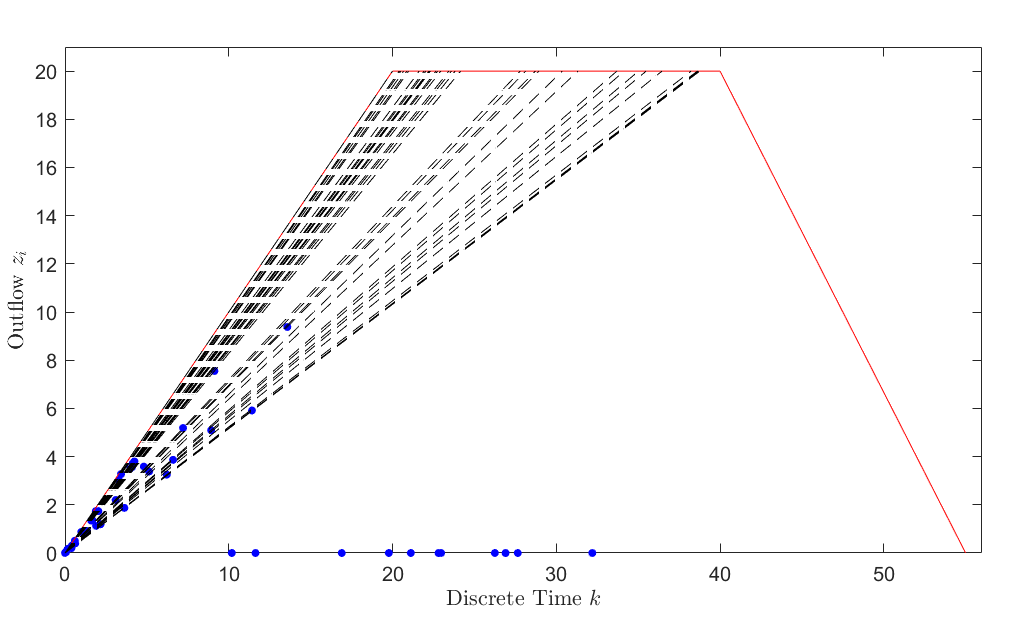}}
\subfigure[$k=30$]{\includegraphics[width=.32\linewidth]{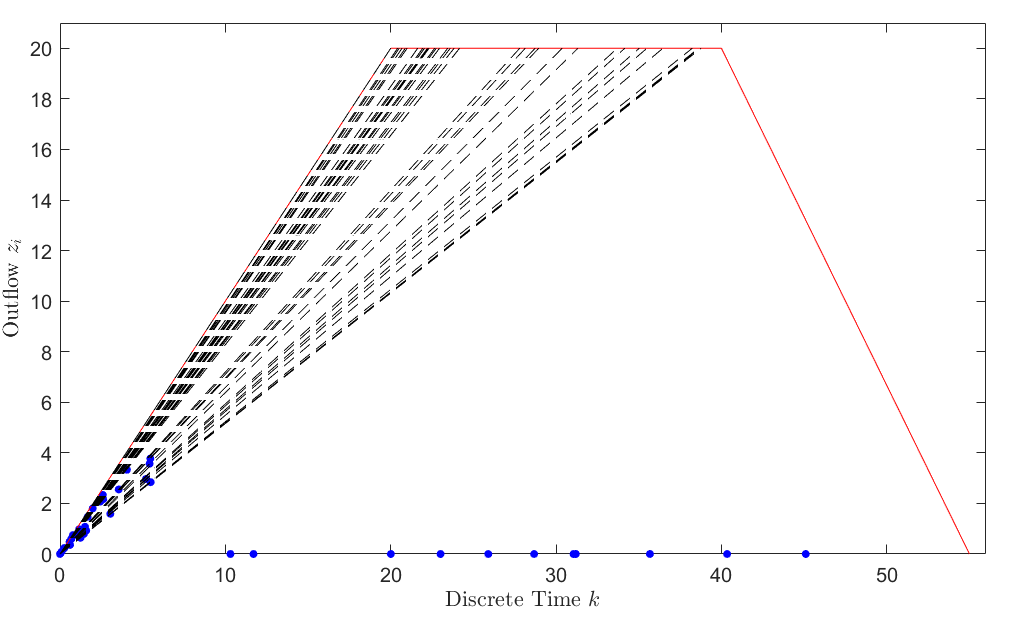}}
\subfigure[$k=50$]{\includegraphics[width=.32\linewidth]{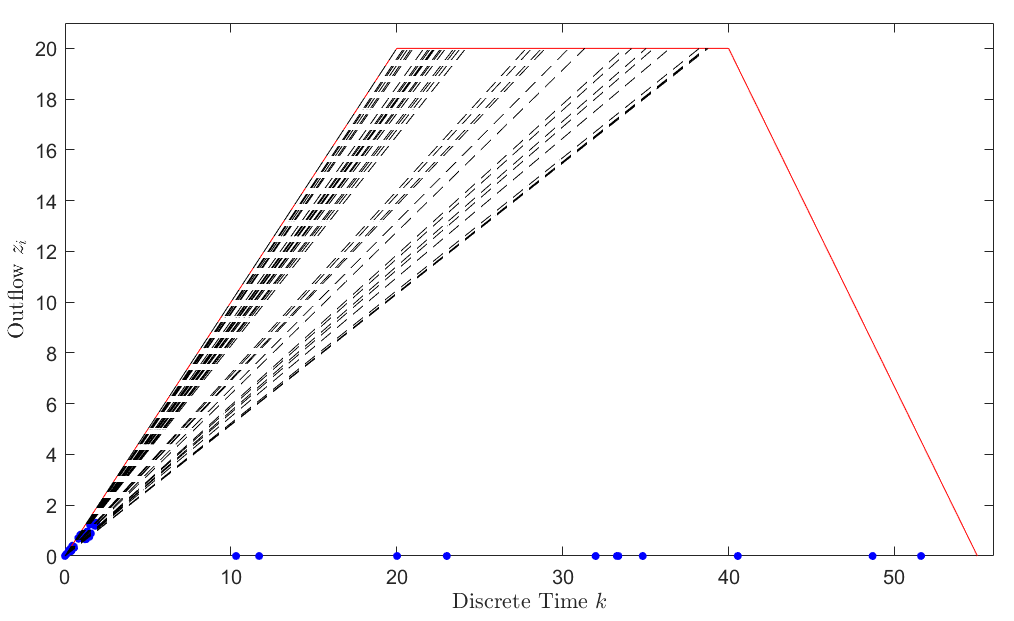}}
\caption{Outflow of the NOIR roads at (a) $k=15$, (b) $k=30$, and (c) $k=50$.}
\label{outflow}
\end{figure*}

\begin{figure}[htb]
\centering
\includegraphics[width=3.0  in]{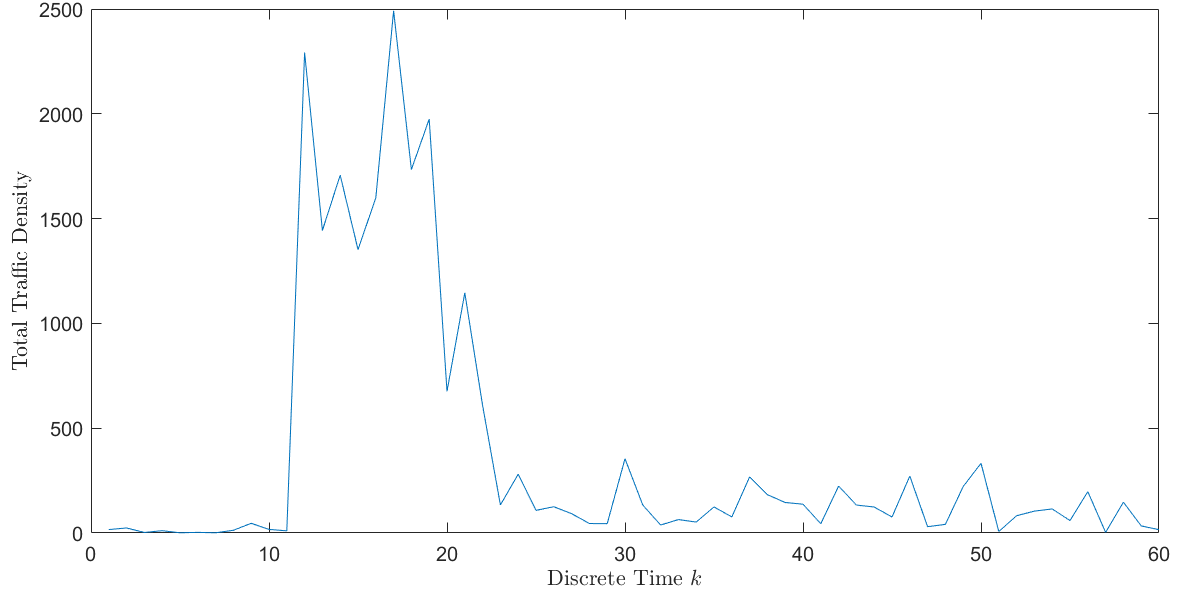}
\caption{Schematix of the cyclic graph $\mathcal{C}_i$ with movement phase repetition. }
\label{totaldensity}
\end{figure}


%
%
For traffic simulation, we consider the same FD to obtain traffic feasibility conditions at every road $i\in \mathcal{W}$ and choose the simulation parameters listed in Table \ref{paramerte}. The optimal boundary inflow is plotted versus discrete time $k$ in Fig. \ref{inflow} for $k=1,\cdots,60$. Fig. \ref{outflow} shows the outflow of $z_i$ of every road road $i\in \mathcal{V}$ at sample times $k=15$, $k=30$, and $k=50$. It is seen that the feasibility conditions imposed by the FD are all satisfied. Also, the net traffic density ($\sum_{i\in \mathcal{V}}\rho_i$) versus discrete time $k$ is plotted in Fig. \ref{totaldensity}.

\IEEEpeerreviewmaketitle

\section{Conclusion}\label{Conclusion}
In this paper, we introduced a  new physics-inspired approach law to model the traffic evolution and control congestion through the boundary roads of a NOIR. By commanding cyclic movement phase rotation at NOIR junctions, we modeled traffic coordination by a switching discrete time dynamics, with deterministic transitions over finite states representing NOIR movement phases. We used a trapezoid FD to formally specify the feasibility and liveness conditions for traffic coordination. The feasibility conditions impose linear equality and inequality constraints on traffic congestion control, which was defined as a receding horizon optimization problem, and can be solved as a quadratic programming problem.



\bibliographystyle{IEEEtran}
\bibliography{references}

\end{document}